\documentclass[authoryear]{elsarticle}
\usepackage{mathtools,upref,siunitx,upquote,fancyvrb,bbm,xspace,color,amsmath,amssymb, bm,amsthm}
\allowdisplaybreaks
\usepackage[hyphens]{url}
\usepackage[utf8]{inputenc}
\usepackage{esdiff}
\usepackage{graphicx}
\usepackage{xcolor}
\usepackage{cleveref}

\DeclareSymbolFont{GreekLetters}{OML}{cmr}{m}{it} %Provide missing letters
\DeclareSymbolFont{UpSfGreekLetters}{U}{cmss}{m}{n} %Provide missing letters
\DeclareMathSymbol{\varrho}{\mathalpha}{GreekLetters}{"25}
\DeclareMathSymbol{\UpSfLambda}{\mathalpha}{UpSfGreekLetters}{"03}
\DeclareMathSymbol{\UpSfSigma}{\mathalpha}{UpSfGreekLetters}{"06}
\providecommand{\mathbold}{\boldsymbol}
\newcommand{\bvec}[1]{\mathbold{#1}}

\newcommand{\dif}{{\rm d}}

\newlength{\overwdth}

\makeatletter
\newcommand*\bigcdot{\mathpalette\bigcdot@{.7}}
\newcommand*\bigcdot@[2]{\mathbin{\vcenter{\hbox{\scalebox{#2}{$\m@th#1\bullet$}}}}}
\makeatother

\def\abs#1{\ensuremath{\left \lvert #1 \right \rvert}}

\newcommand{\norm}[2][{}]{\ensuremath{\left \lVert #2 \right \rVert}_{#1}}

\DeclareMathOperator{\opt}{opt}

\DeclareMathOperator{\Order}{{\mathcal O}}

\newcommand{\vzero}{\bvec{0}}
\newcommand{\vone}{\bvec{1}}

\newcommand{\vh}{\bvec{h}}

\newcommand{\vk}{\bvec{k}}

\newcommand{\vw}{\bvec{w}}

\newcommand{\vx}{\bvec{x}}
\newcommand{\vX}{\bvec{X}}

\newcommand{\vz}{\bvec{z}}

\newcommand{\vgamma}{\bvec{\gamma}}

\newcommand{\vDelta}{\bvec{\Delta}}

\newcommand{\tf}{\widetilde{f}}

\newcommand{\mK}{\mathsf{K}}

\newcommand{\bbone}{\mathbbm{1}}
\newcommand{\naturals}{\mathbb{N}}
\newcommand{\reals}{\mathbb{R}}
\newcommand{\integers}{\mathbb{Z}}
\newcommand{\natzero}{\mathbb{N}_{0}}

\newcommand{\cf}{\mathcal{F}}

\newcommand{\cl}{\mathcal{L}}

\newcommand{\me}{\ensuremath{\mathrm{e}}} % for math number 'e', 2.718 281 8..., tha base of natural logarithms
\usepackage{algpseudocode}
\usepackage{algorithm, algorithmicx}
\algnewcommand\algorithmicparam{\textbf{Parameters:}}
\algnewcommand\PARAM{\item[\algorithmicparam]}
\algnewcommand\algorithmicinput{\textbf{Input:}}
\algnewcommand\INPUT{\item[\algorithmicinput]}
\algnewcommand\RETURN{\State \textbf{Return }}

\newcommand{\tr}{\widetilde{r}}
\newcommand{\appxintn}{\appxint_n}
\DeclareMathOperator{\appxint}{\hat{I}}
\DeclareMathOperator{\trun}{trunc}
\newcommand{\onetos}{1\!:\!s}

\DeclareMathOperator{\Bern}{bern}

\newtheorem{theorem}{Theorem}

\begin{document}
\title{The Quality of Lattice Sequences for \\ Arbitrary Sample Size}
\author[inst1]{Larysa Matiukha\corref{cor1}}
\ead{larysa.matiukha@gmail.com}
\cortext[cor1]{Corresponding author}
\author[inst1]{Yuhan Ding}
\author[inst1]{Fred J. Hickernell}
\affiliation[inst1]{
organization = {Illinois Institute of Technology, Chicago, IL 60616, USA},
  city={Chicago},
  state={IL},
  country={USA}
}
\begin{abstract} Extensible lattices are a common choice of 
nodes for approximating multidimensional integrals by a sample mean, typically using sample sizes of the form $n = b^m$ for a prime base $b$ and non-negative integer $m$. However, a computational time budget or early termination due to hardware failure may prevent us from choosing the preferred number of samples. In this paper, we establish an upper bound on the figure of merit $P_\alpha$ for extensible lattice sequences \emph{with arbitrary $n$}, derived as a worst-case error in a Banach space setting. 
We show that although the upper bound on $P_\alpha$ can decay no faster than $\Order(n^{-1})$ for general $n$, restricting $n$ to the form $n =\lambda b^p$, where $\lambda$ is a fixed integer, we can recover the optimal decay of nearly $\Order(n^{-\alpha})$ as $p \to \infty$. We also investigate a related figure of merit, $P_{\alpha,2}$, arising as a worst-case error for Hilbert spaces of integrands. We obtain a theoretical upper bound on $P_{\alpha,2}$ in terms of $P_{\alpha}$ for equal sample weights and show numerically that optimal sample weights can improve the cubature error bound even further.

\end{abstract}

\maketitle

\section{Introduction}
Multidimensional integrals arise in various applications in quantitative finance \citep{CafMorOwe97}, statistics \citep{Gen92,GenBre99}, and physics \citep{Kei96,PapTra97}, where they are often expressed as expectations of random variables.
Lattices are a popular choice of nodes $\{\vz_i\}_{i=0}^{\infty} \in [0,1)^s$ for approximating multidimensional integrals by a sample mean,
\[
\int_{[0,1)^s} f(\vx) \, \dif \vx \approx \frac{1}{n} \sum_{i=0}^{n-1} f(\vz_i) =:\appxintn(f),
\]
\citep{DicEtal22a,Nie92,SloJoe94}.

\cite{HicNie03a} established the existence of extensible lattices with good generating vectors for the preferred values $n = b, b^2, \ldots$. 
In practice, it might not always be possible to restrict the sample size to $n = b^m$ due to early termination or computational time budget. Motivated by this, we aim to extend these results to all positive $n$, recognizing that the upper bounds on the figures of merit will be somewhat worse for general $n$ than for powers of the base.

The remainder of the paper is organized as follows. Section 2 reviews the background on integration lattices and existing results. Section 3 introduces the figure of merit $P_\alpha$ for arbitrary $n$, an extension of the traditional $P_\alpha$ defined for full lattices.  This $P_\alpha$ is the worst-case error for integrands in a Korobov (Banach) space indexed by the smoothness parameter $\alpha$.  Section 4 establishes the upper bound on $P_\alpha$ when $n$ is not a power of the base. Section 5 considers a related quantity, $P_{\alpha,2}$, which is the worst case cubature error for the analogous Hilbert space.  The figure of merit $P_{\alpha,2}$ is more amenable to calculation than $P_\alpha$, and there is a discussion of optimal sample weights and worst case cubature error for arbitrary $n$.

\section{Background} % main things + previous results 
Historically, lattice points were constructed as sets with fixed cardinality, $n$, and took the form
\begin{equation*} %\label{eq:lat}
    \{\vz_i = i \vh/n \bmod{\vone} : i=0,1, \ldots, n-1 \} \in [0,1)^s,
\end{equation*}
where $\vh \in \{1, \ldots, n-1\}^s$ is the \emph{generating vector}. Note that this set is closed under addition modulo $\vone$. A (random) shift, $\vDelta \in [0,1)^s$, is often added:
\begin{equation*}% \label{eq:shlat}
    \{\vz_i = i \vh/n + \vDelta \pmod{\vone} : i=0,1, \ldots, n-1 \} \in [0,1)^s.
\end{equation*}

Since this construction depends on $n$, a new point set must be generated each time the sample size changes. To address this, lattices were later defined extensibly, i.e., not for a fixed $n$, allowing the lattice to be extended without discarding previously generated points. Extensible lattice sequences, proposed in \cite{Mai81a} and \cite{HicEtal00}, take the form 
\begin{equation*} \label{eq:extlat}
    \{\vz_i = \vh\phi(i)+ \vDelta \pmod{\vone} : i=0,1, \ldots \} \in [0,1)^s,
\end{equation*}
where $\{\phi(\cdot)\}_{i=0}^\infty$ is the van der Corput sequence in base $b$.  (The base is suppressed in our notation for simplicity.) In this case $\vh$ must be a generalized integer as defined in \cite[Section 2]{HicNie03a}. An example of a two-dimensional shifted extensible lattice is shown in \Cref{fig:lattice-plot}. 

\begin{figure}[H]
\centering
\includegraphics[width=11cm]{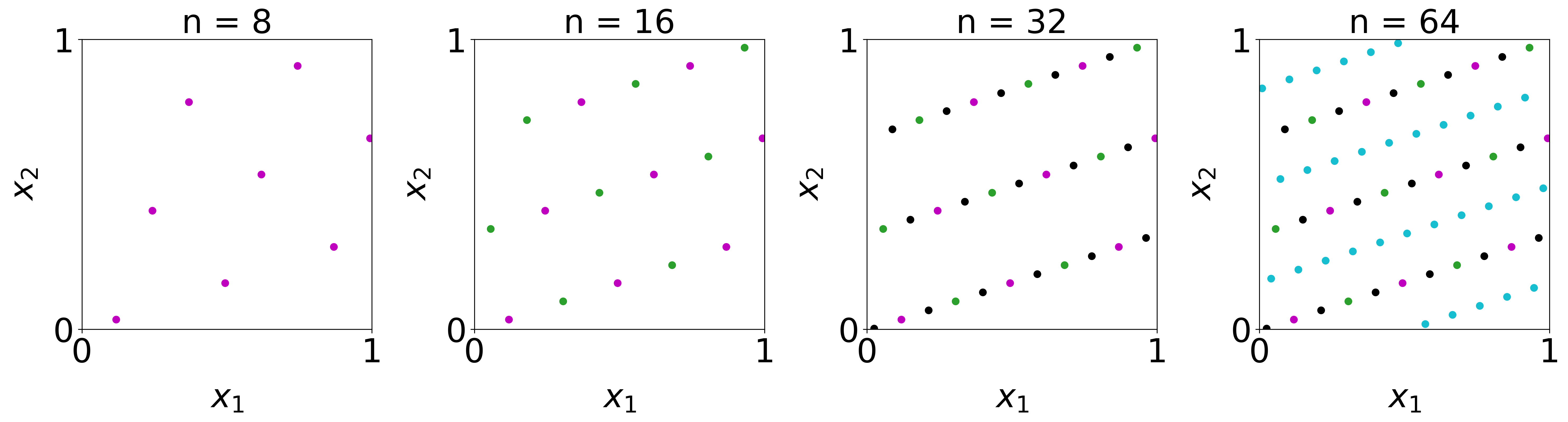}
\caption{The first 64 points of two-dimensional shifted, extensible lattice. As $n$ increases, the square becomes more uniformly filled.}
\label{fig:lattice-plot}
\end{figure}

The van der Corput sequence is defined as: 
\[
\phi((\cdots i_2 i_1 i_0)_b) = {}_b0.i_0 i_1 i_2 \cdots.
\]
For example, for $b=2$,
\[
\phi(6) = \phi(110_2) = {}_20.011 = \frac 38.
\]
Note that the first ${b^m}$ points of the van der Corput sequence are just equally spaced points reordered. 
\begin{equation} \label{eq:phipropone}
\{ \phi(i) : i = 0, \ldots, b^m-1 \} = \{0, b^{-m}, 2\times b^{-m}, \ldots, 1 - b^{-m} \}.
\end{equation}
Also note that the set of $b^m$ points starting at $\lambda b^m$ for positive integer $\lambda$ is a shift by $\phi(\lambda \times b^m)$ of the first $b^m$ points:
\begin{multline} \label{eq:phiproptwo}
\{ \phi(i) : i = \lambda \times b^m , \ldots, (\lambda+1)b^m-1 \} \\
= \{\phi(\lambda \times b^m) + 0, \phi(\lambda \times b^m) + b^{-m}, \ldots, \phi(\lambda \times b^m) + 1 - b^{-m} \} , \\
\lambda \in \natzero.
\end{multline}

\cite{HicNie03a} provided bounds on the worst-case cubature error using the first $n= b, b^2, \ldots$ points as sampling sites. Their result is quoted below in \eqref{eq:Niedbd}. The case of using the first \emph{arbitrary} $n$ nodes of a lattice sequence for cubature in a Hilbert space was discussed in \cite{HicEtal10a}[Corollary 1].  There it was shown that if a lattice sequence in base $b$ has the worst case error, no greater than $C(b^m)^{-\alpha}$ for the first $b^m$ points, for $m=0, 1, 2, \ldots$, for some constants $C, \alpha > 0$, then the worst case error of using the first $n \geq 2$ points of the lattice sequence is no greater than $ C \log_b(n) \, {n^{-\min(\alpha,1)}}$. 

We extend the bounds on $P_\alpha$ given in \cite{HicNie03a} to arbitrary $n$ in \Cref{thm:one}.  Moreover, in the Hilbert space setting considered by \cite{HicEtal10a}, we derive an error bound that does not include their additional logarithmic term for $\alpha > 1$ in \Cref{thm:three}.

\section{Worst case error analysis for arbitrary $n$}
This section follows the traditional error analysis presented in \cite{DicEtal22a}. However, we extend the theory to the case where the nodeset is not a full lattice. 
We first consider integrands, $f$, that have an absolutely summable Fourier series:
\begin{equation*} %\label{eq:fseries}
    f(\vx) = \sum_{\vk \in \integers^{s}} \tf(\vk) \me^{2 \pi \sqrt{-1} \vk^T \vx}, \qquad \text{where } \tf(\vk) = \int_{[0,1)^{s}} f(\vx) \me^{-2 \pi \sqrt{-1} \vk^T \vx}\, \dif \vx.
\end{equation*}
Given the positive coordinate weights, $\gamma_1$, $\gamma_2, \ldots$, we define a Banach space of integrands by weighting the Fourier coefficients: 
\begin{gather}
\nonumber
\tr(\vk,\vgamma) = \prod_{j=1}^{s} r(k_{j},\gamma_{j}),
\qquad \text{where } r(k_{j},\gamma_{j})=\begin{cases} 1, &
k_{j}=0, \\ \gamma_{j}^{-1}\abs{k_{j}}, & k_{j} \ne 0,  \end{cases}
\\
\label{eq:Banachspace}
\cf_{\alpha} = \{ f \in \cl_2[0,1)^{s} :
\norm[\cf{\alpha}]{f} < \infty \}, \qquad
\norm[\cf{\alpha}]{f} := \sup_{\vk \in \integers^{s}}
\left(\tr(\vk,\vgamma)^{\alpha} \abs{\tilde{f}(\vk)} \right).
\end{gather}
Here, the parameter $\alpha$ denotes the smoothness of the integrands.

It then follows that the worst-case error in approximating the integral by the sample mean using arbitrary point set is
\begin{align} \label{eq:wcerrPalpha}
\nonumber
\MoveEqLeft{\sup_{\norm[\cf_{\alpha}]{f} \le 1} \abs{\int_{[0,1)^{s}} f(\vx) \, \dif \vx - \appxint_n(f)}} \\
&
\nonumber
= \sup_{\norm[\cf_{\alpha}]{f} \le 1} \abs{\sum_{\vk \in \integers^{s}} \tf(\vk) \left[\int_{[0,1)^{s}} \me^{2 \pi \sqrt{-1} \vk^T \vx} \, \dif \vx - \appxintn(\me^{2 \pi \sqrt{-1} \vk^T \cdot})\right]} \\
& \nonumber= \sup_{\norm[\cf_{\alpha}]{f} \le 1} \abs{\sum_{\vk \in \integers^s \setminus\{\vzero\}} \tf(\vk) \left[ \appxintn(\me^{2 \pi \sqrt{-1} \vk^T \cdot})\right]} \\
& = \underbrace{\sum_{\vk \in \integers^s \setminus\{\vzero\}} \abs{\appxint_n(\me^{2 \pi \sqrt{-1} \vk^T \cdot})}\tr(\vk,\vgamma)^{-\alpha}}_{=: P_\alpha(\vh,\vgamma,n,s) = \text{quality of the lattice nodes}}, \\
\nonumber\end{align}
where we take advantage of H\"older's inequality and note that the error term vanishes for $\vk = \vzero$, and the integral itself vanishes for all $\vk \ne \vzero$.
This notation follows \cite[(4)]{HicNie03a} but we 
%replace 1:s in \cite.. by s, 2 bullets
\begin{itemize}
    \item Do not assume the first $n$ points form a lattice and
    \item Replace their $\onetos$ by $s$.
\end{itemize}
For an extensible lattice with $n = b^m$, the quality measure $P_\alpha(\vh,\vgamma,n, s)$ takes on a simple form.  Note that
\begin{align*}\label{Eq:DualLattice}
    \MoveEqLeft{\appxint_{b^m}(\me^{2 \pi \sqrt{-1} \vk^T \cdot})} \\
    \nonumber
    &= \frac 1{b^m} \sum_{i=0}^{b^m-1} \me^{2 \pi \sqrt{-1} \vk^T (\vh\phi(i) + \vDelta \bmod{\vone})} \\
    \nonumber
    &= \frac {\me^{2 \pi \sqrt{-1} \vk^T \vDelta}}{b^m} \sum_{i=0}^{b^m-1} \me^{2 \pi \sqrt{-1} \vk^T \vh\phi(i)} \\
    \nonumber
    &= \frac {\me^{2 \pi \sqrt{-1} \vk^T \vDelta}}{b^m} \sum_{i=0}^{b^m-1} \me^{2 \pi \sqrt{-1} \vk^T \vh i/b^m} \qquad \text{by \eqref{eq:phipropone}}\\
    \nonumber
    &= \frac {\me^{2 \pi \sqrt{-1} \vk^T \vDelta}}{b^m} \times
    \begin{cases}
    \frac{\me^{2 \pi \sqrt{-1} \vk^T \vh} - 1}{\me^{2 \pi \sqrt{-1} \vk^T \vh/b^m} - 1} = 0, & \vk^T \vh \pmod{b^m} \ne 0\\
    b^m, & \vk^T \vh \pmod{b^m} = 0
    \end{cases} \\
    \nonumber
    & = \me^{2 \pi \sqrt{-1} \vk^T \vDelta} \bbone_{B(\vh,m,s)}(\vk).
\end{align*}
Thus it follows that
\begin{equation} \label{eq:Palphadual}
    P_\alpha(\vh,\vgamma,b^m,s) = \sum_{\vk \in B(\vh,m,s)} \tr(\vk,\vgamma)^{-\alpha},
\end{equation}
where $B(\vh,m,s) : = \{\vk \in  \integers^s \setminus \{\vzero\} : \vk^T \vh \pmod{b^m} = 0\}$ is the \emph{dual lattice} minus the origin.  This corresponds to \cite[(3)]{HicNie03a}.  

Although \eqref{eq:Palphadual} is an infinite sum, $P_{2\alpha}$ for positive integer $\alpha$ can be converted to a finite sum in terms of Bernoulli polynomials:
\begin{equation} \label{eq:P2alpha}
P_{2\alpha}(\vh,\vgamma,b^m,s) = -1 + \frac{1}{n}\sum_{i = 0}^{n-1}  
\prod_{\ell = 1}^s [1 + \tilde{\gamma}_{\ell}^{2\alpha}\Bern_{2\alpha}(x_{i,\ell})],
\end{equation}
where 
\begin{equation}  \label{eq:gammatilde}
    \tilde{\gamma}_{\ell}^{2\alpha} = \frac{(2\pi\gamma_{\ell} )^{2\alpha}}{(-1)^{\alpha +1}(2\alpha)!}.
\end{equation}
This follows because according to \cite[Equation 24.8.3]{OlvEtal10a}, the even Bernoulli polynomials have an expansion of 
\[
\Bern_{2\alpha}(x) = \frac{(-1)^{\alpha+1} (2\alpha)!}{(2 \pi)^{2\alpha}} \sum_{k \in \mathbb{Z}, \ k \ne 0} \frac{\exp(2\pi\sqrt{-1} k x)}{k^{2\alpha}}, \quad 0 \le x \le 1, \ \alpha = 1, 2, \ldots.
\]
However, for $n \ne b^m$ such a finite sum does not exist.

We know by \cite[Theorem 5]{HicNie03a} that  there exists $\vh$ with
\begin{multline} \label{eq:Niedbd}
    P_{\alpha}(\vh,\vgamma,b^m,s) \le C_{P}(b,\alpha,\vgamma,\epsilon,s)
    b^{-m\alpha} (\log b^{m})^{\alpha(s+1)} [\log \log (
    b^m+1)]^{\alpha(1+\epsilon)}, \\ m = 1, 2,\ldots, \quad \alpha \ge 1.
\end{multline}
We emphasize that the generating vector, $\vh$ does not depend on $m$.

This result can be extended to the case when $m = 0$ by noting that
\begin{align*}
    \MoveEqLeft{P_{\alpha}(\vh,\vgamma,1,s)} \\
    & = \sum_{\vk \ne \vzero} \tr(\vk, \vgamma)^{-\alpha} = -1 + \sum_{\vk \in \integers^s} \tr(\vk,\vgamma)^{-\alpha} \\
    & = -1 + \sum_{k_1 =-\infty}^{\infty} \cdots \sum_{k_s  =-\infty}^{\infty} \frac{1}{\prod_{j=1}^s[\max(1,\gamma_j^{-1}\abs{k_j})]^\alpha} \\
    & = -1 + \sum_{k_1 =-\infty}^{\infty} \frac{1}{[\max(1,\gamma_1^{-1}\abs{k_1})]^\alpha} \cdots \sum_{k_s  =-\infty}^{\infty} \frac{1}{[\max(1,\gamma_s^{-1}\abs{k_s})]^\alpha} \\
    & = -1 + \prod_{j=1}^s \left[\sum_{k =-\infty}^{\infty} \frac{1}{[\max(1,\gamma_j^{-1}\abs{k})]^\alpha} \right] \\
    & = -1 + \prod_{j=1}^s \left[1 + 2 \gamma_j^{\alpha} \sum_{k =1}^{\infty} \frac{1}{k^\alpha} \right]\\
    & = -1 +  \prod_{j=1}^s \left[1 + 2 \gamma_j^{\alpha}\zeta(\alpha) \right],
\end{align*}
where $\zeta$ is the Riemann zeta function.

Therefore, the above formula \eqref{eq:Niedbd} can be extended to the case of $m=0$ as follows:
\begin{multline} \label{eq:Palphaextm}
    P_{\alpha}(\vh,\vgamma,b^m,s) \\
    \le C_{P}(b,\alpha,\vgamma,\epsilon,s)
    b^{-m\alpha}\max(1, (\log b^{m})^{\alpha(s+1)})[\max(1,\log \log (
    b^m+1))]^{\alpha(1+\epsilon)}, \\ m =0, 1, 2,\ldots, \quad \alpha > 1.
\end{multline} 
In the next section, we extend this upper bound to all positive $n$. 

\section{The figure of merit $P_{\alpha}(\vh,\vgamma,n,s)$ for  $n$ other than $b^m$} \label{sec:Palphabd}
Starting from the upper bound for $P_{\alpha}(\vh,\vgamma,n,s)$ for $n$ a non-negative power of the base, $b$, it is possible to obtain an analogous result arbitrary $n$: 

\begin{theorem} \label{thm:one}
    For fixed integer base $b \ge 2$, a fixed smoothness parameter $\alpha > 1$, a fixed vector of coordinate weights $\vgamma \in [0,\infty)^\infty$, and a fixed $\epsilon > 0$, 
    
    \begin{enumerate}
    \renewcommand{\labelenumi}{\alph{enumi}.}
        \item For arbitrary $n \in \naturals$,
        
        \begin{enumerate}
        \renewcommand{\labelenumii}{\roman{enumii}.}
            \item There exists a generating vector for the lattice, $\vh$, for which the figure of merit, $P_\alpha$, has the following upper bound for arbitrary positive integer $n$:
    \begin{multline} \label{eq:mainresultone}
        P_\alpha(\vh,\vgamma,n,s) 
        \le \frac {C_{P}(b,\alpha,\vgamma,\epsilon,s)}{n} \frac{(b-1)}{1 - b^{1-\alpha}} \\
        \times 
  \max(1, (\log n)^{\alpha(s+1)})
      [\max(1,\log \log (n+1))]^{\alpha(1+\epsilon)} \\
      \forall n , s \in \naturals.
    \end{multline}
    This is the same  constant, $C_{P}(b,\alpha,\vgamma,\epsilon,s)$, and the same power of logarithm as in \eqref{eq:Palphaextm} for $n = b^{m}$, but the power of $n$ is fixed at $-1$.

    \item The powers of log may be eliminated in favor of a larger constant, as follows:
    \begin{equation} \label{eq:mainresultoneb}
        P_\alpha(\vh,\vgamma,n,s) 
        \le \frac {C_{P}(b,\alpha,\vgamma,\epsilon,s) S(b, \alpha, \epsilon, s)}{n} \qquad
      \forall n , s \in \naturals,
    \end{equation}
    where the constant $S(b, \alpha, \epsilon, s)$ is defined in \eqref{eq:infseries}.

    \end{enumerate}

    \item However, if $n = \lambda b^p$ for some positive integer $p$ and some positive integer  $\lambda$ that is relatively prime with respect to $b$, then, %
    $P_\alpha$ has an upper bound that is similar to the case of $n=b^m$:
    \begin{multline} \label{eq:mainresulttwo}
         P_\alpha(\vh,\vgamma,n,s)
         \le \frac{{C}_{P}(b,\alpha, \vgamma,\epsilon,s)}{n^{\alpha}}  \frac {\lambda^{\alpha -1}(b-1)} {1 - b^{1- \alpha}}  \\
         \times \max(1, (\log n )^{\alpha(s+1)})) [\max(1,\log \log (
    n+1))]^{\alpha(1+\epsilon)} \\
 \ n, \lambda, s \in \naturals.
    \end{multline} 
    Again, the constant $C_{P}(b,\alpha,\vgamma,\epsilon,s)$ and the powers of the logarithms are the same as in \eqref{eq:Palphaextm}. Now, the power of $n$ is similar to that in \eqref{eq:Palphaextm}, but there is an extra factor of $\lambda^{\alpha -1}(b-1)/(1 - b^{1-\alpha})$.  Thus it is especially useful for small $\lambda$. (The parameter $p = \log_b(n/\lambda)$ does not appear above because it can be written in terms of $n$ and $\lambda$.)

     \end{enumerate}
    
\end{theorem}

Before giving the proof of this theorem, we note that 

\begin{itemize}
    \item In \eqref{eq:mainresultone} the logarithmic terms are the same as for $n=b^m$, but the power of $n$ is only $-1$ for all $\alpha > 1$.  This degradation is inevitable as noted in \cite{HicEtal10a} and elsewhere.
    \item In \eqref{eq:mainresulttwo} the dependence of the upper bound $P_\alpha$ on $n$ is the same as for $n=b^m$ (except for multiplicative constants), but  here it is implicitly assumed that $n$ is a small multiple of $b^m$.  Otherwise, the $\lambda^{\alpha-1}$ term might be substantial.
\end{itemize}

\subsection{Proof of \eqref{eq:mainresultone}}
The key tool in this proof is the triangle inequality.  A careful accounting of terms keeps the upper bound reasonably conservative.

Let any non-negative integer $n$ be expressed as  its $b$-ary expansion, i.e., $n = n_0 + bn_1 + \cdots + b^m n_m$, where the digits $n_0, n_1, \ldots$ are in $\{0, 1, \cdots, b-1\}$, and $n_m > 0$. Thus, $m = \lfloor \log_b(n) \rfloor$.  Let $\trun(n,\ell) := n_0 + \cdots + b^{\ell} n_{\ell}$ for $\ell = 0, \ldots, m$ and $\trun(n,-1) := 0$.  Then the lattice rule approximation to the integral of the complex exponential function may be expressed in terms of a sum involving whether the wavenumber lies in the dual lattice:  
\begin{align*}
    \MoveEqLeft{\appxint_{n}(\me^{2 \pi \sqrt{-1} \vk^T \cdot}) } \\
    &= \frac 1{n} \sum_{i=0}^{n-1} \me^{2 \pi \sqrt{-1} \vk^T (\vh\phi(i) + \vDelta \pmod{\vone})} \\
    \nonumber
    &= \frac {\me^{2 \pi \sqrt{-1} \vk^T \vDelta}}{n} \sum_{\ell = 0}^{m} \,
    \sum_{i=n - \trun(n,m-\ell) }^{n - \trun(n,m-\ell-1)  -1} \me^{2 \pi \sqrt{-1} \vk^T \vh\phi(i)} \\
    &= \frac {\me^{2 \pi \sqrt{-1} \vk^T \vDelta}}{n} \sum_{\ell = 0}^{m} \, n_{m-\ell}
    \sum_{i=0 }^{b^{m-\ell}  -1} \me^{2 \pi \sqrt{-1} \vk^T \vh[i/b^m + \phi(n - \trun(n,m-\ell))]} \\
    & \hspace{30ex} \text{by \eqref{eq:phiproptwo}} \\
    & = \frac {\me^{2 \pi \sqrt{-1} \vk^T \vDelta}}{n} \sum_{\ell = 0}^{m} b^{m-\ell} n_{m - \ell} \bbone_{B(\vh,m-\ell,s)}(\vk) \me^{2 \pi \sqrt{-1} \vk^T \vh\phi(n - \trun(n,m-\ell))}.
\end{align*}
 The absolute value of this lattice rule cubature is then bounded above by the triangle inequality as 
\begin{equation} \label{eq:ihatarbn}
    \abs{\appxint_{n}(\me^{2 \pi \sqrt{-1} \vk^T \cdot})} \le \frac {1}{n} \sum_{\ell = 0}^{m} b^{m-\ell} n_{m - \ell} \bbone_{B(\vh,m-\ell, s)}(\vk).
\end{equation}

The figure of merit, $P_\alpha(\vh,\vgamma,n,s)$, depends on how small these quantities can be made by a good choice of $\vh$, namely,
\begin{align} \label{eq:Palphan}
      \nonumber
      P_\alpha(\vh,\vgamma,n,s) & = \sum_{\vk \in \integers^s \setminus\{\vzero\}} \abs{\appxint_n(\me^{2 \pi \sqrt{-1} \vk^T \cdot})}\tr(\vk,\vgamma)^{-\alpha} \qquad \text{by \eqref{eq:wcerrPalpha}} \\
      \nonumber
      & \le \frac {1}{n} \sum_{\ell = 0}^{m} b^{m-\ell} n_{m - \ell} \sum_{\vk \in B(\vh,m-\ell,s)} \tr(\vk,\vgamma)^{-\alpha}
      \qquad \text{by \eqref{eq:ihatarbn}}\\
      & = \frac {1}{n} \sum_{\ell = 0}^{m} b^{m-\ell} n_{m - \ell} P_\alpha(\vh,\vgamma,b^{m-\ell},s)
      \qquad \text{by \eqref{eq:Palphadual}}.
\end{align}
Note that $P_\alpha$ for arbitrary $n$ has now been bounded above by a sum of $P_\alpha$ for powers of the base.  

The next part of the proof proceeds by substituting in the upper bounds for these $P_\alpha$ from \eqref{eq:Palphaextm}:

\begin{align} \label{eq:sameargument}
      \nonumber
      \MoveEqLeft
      P_\alpha(\vh,\vgamma,n,s) \\ & 
      \nonumber 
       \le  \frac {1}{n} \left\{ b^m n_m P_\alpha(\vh,\vgamma,b^m,s) + b^{m-1} n_{m-1} P_\alpha(\vh,\vgamma,b^{m-1},s) \right . \\
       \nonumber
       & \left . \qquad + \cdots + b^0 n_0 P_\alpha(\vh,\vgamma,b^0,s)  \right\} \\
       \nonumber
       &\le \frac {(b-1)}{n} \left\{ b^m  P_\alpha(\vh,\vgamma,b^m,s) + b^{m-1}  P_\alpha(\vh,\vgamma,b^{m-1},s) \right . \\
       \nonumber
       & \qquad \left . + \cdots + b^0  P_\alpha(\vh,\vgamma,b^0,s)  \right\} 
      \qquad \text{since $n_m \le b-1$}\\
      \nonumber
      & \le \frac {(b-1)C_{P}(b,\alpha,\vgamma,\epsilon,s)}{n} \\
      \nonumber
      & \qquad \times \left\{ b^{m(1-\alpha)}\max(1, (\log b^{m})^{\alpha(s+1)}) [\max(1,\log \log (
    b^m+1))]^{\alpha(1+\epsilon)}\right .\\ 
    \nonumber
      & \qquad  + b^{(m-1)(1-\alpha)}\max(1, (\log b^{(m-1)})^{\alpha(s+1)}) [\max(1,\log \log (
    b^{(m-1)}+1))]^{\alpha(1+\epsilon)}\\
      \nonumber
      & \qquad + \cdots \\
       \nonumber
      & \qquad  \left . + b^{-0}\max(1, (\log b^{0})^{\alpha(s+1)}) [\max(1,\log \log (
    b^{0}+1))]^{\alpha(1+\epsilon)} \right\}\\
      & \qquad \qquad \qquad \alpha > 1 \qquad \text{by \eqref{eq:Palphaextm}} .
\end{align}

Finally, we factor out the largest logarithmic factors and sum the geometric series:
\begin{align*}
\MoveEqLeft P_\alpha(\vh,\vgamma,n,s) \\
      & \le \frac {(b-1)C_{P}(b,\alpha,\vgamma,\epsilon,s)}{n} \frac{(1 - b^{(m+1)(1-\alpha)})}{1 - b^{1-\alpha}}
      \max(1, (\log b^{m})^{\alpha(s+1)}) \\
      & \qquad \qquad \times [\max(1,\log \log (
    b^m+1))]^{\alpha(1+\epsilon)} \\
      & \le \frac {C_{P}(b,\alpha,\vgamma,\epsilon,s)}{n} \frac{(b-1)}{1 - b^{1-\alpha}}
  [\max(1, (\log n)^{\alpha(s+1)})
      [\max(1,\log \log (n+1))]^{\alpha(1+\epsilon)} 
\end{align*}
This completes the proof of \eqref{eq:mainresultone}. 
% mention Hickernell&friends result
Unfortunately, the upper bound decays only like $\Order(n^{-1+\delta})$.  This is due to the $m = 0$ term in the sum above.  This suggests our later result in \eqref{eq:mainresulttwo} that if the small powers of $b$ vanish, we can get a bound with a better power of $n$.

\subsection{Proof of \eqref{eq:mainresultoneb}}  We may remove the logarithmic factors in favor of a larger constant.  Since $\alpha > 1$, the infinite series 
\begin{multline} \label{eq:infseries}
    S(b, \alpha, \epsilon, s) : = b^{0}\max(1, (\log b^{0})^{\alpha(s+1)}) [\max(1,\log \log (
    b^{0}+1))]^{\alpha(1+\epsilon)} \\
    + b^{1(1-\alpha)}\max(1, (\log b^{1})^{\alpha(s+1)}) [\max(1,\log \log (
    b^{1}+1))]^{\alpha(1+\epsilon)}
    + \cdots \\
    + b^{m(1-\alpha)}\max(1, (\log b^{m})^{\alpha(s+1)}) [\max(1,\log \log (
    b^m+1))]^{\alpha(1+\epsilon)} + \cdots 
\end{multline}
converges.  This implies \eqref{eq:mainresultoneb} by \eqref{eq:sameargument}.

\subsection{Proof of \eqref{eq:mainresulttwo}}
Next, suppose $n = \lambda b^p$, where $\lambda$ is an integer that is relatively prime with respect to $b$.  The proof proceeds similarly to that in the previous subsection, but we obtain a more favorable power of $n$.

The  $b$-ary expansion of $n$ takes the form:  $n = b^pn_p + b^{p+1}n_{p+1} + \cdots + b^m n_m$. 
Since the lattice rule cubature is bounded above \eqref{eq:ihatarbn}, it follows that $P_\alpha$ for $n = \lambda b^p$ can also be bounded above by a sum of $P_\alpha$ for powers of the base: 
\begin{align*}
    P_\alpha(\vh,\vgamma,\lambda b^p,s)
    & = \sum_{\vk \in \integers^s \setminus\{\vzero\}} \abs{\appxint_{\lambda b^p}(\me^{2 \pi \sqrt{-1} \vk^T \cdot})}\tr(\vk,\vgamma)^{-\alpha} \qquad \text{by \eqref{eq:wcerrPalpha}} \\
    & \le \frac {1}{\lambda b^p} \sum_{\ell = 0}^{m-p} b^{m-\ell} n_{m -\ell} P_\alpha(\vh,\vgamma,b^{m-\ell},s) \qquad \text{by \eqref{eq:Palphadual}}
\end{align*}
Similarly to the proof for arbitrary $n$, we substitute in the upper bounds for $P_\alpha$:  
\begin{align} \label{eq:sameargument1}
    \nonumber
    \MoveEqLeft[1]
   P_\alpha(\vh,\vgamma,\lambda b^p,s)  \\
    &\le \frac {1}{\lambda b^p} \left\{ b^m n_m P_\alpha(\vh,\vgamma,b^m,s) + b^{m-1} n_{m-1} P_\alpha(\vh,\vgamma,b^{m-1},s) \right. \\
     \nonumber
    &\qquad + \cdots + b^p n_p P_\alpha(\vh,\vgamma,b^p,s) 
    \bigg\} \\ 
     \nonumber
    & \le \frac {(b-1)}{\lambda b^p}  
    \left\{
    b^m  P_\alpha(\vh,\vgamma,b^m,s) + b^{m-1}  P_\alpha(\vh,\vgamma,b^{m-1},s) \right. \\
     \nonumber
    &\qquad + \cdots + b^p  P_\alpha(\vh,\vgamma,b^p,s) 
    \bigg\} \\
     \nonumber
    & \le \frac {(b-1)}{\lambda b^p} C_{P}(b,\alpha,\vgamma,\epsilon,s) \\
     \nonumber
    &\qquad \times \left\{ b^{m(1-\alpha)}\max\left(1, (\log b^m)^{\alpha(s+1)}\right) \max\left(1,\log \log (b^m+1)\right)^{\alpha(1+\epsilon)} \right.\\ 
     \nonumber
    & \qquad  + b^{(m-1)(1-\alpha)}\max\left(1, (\log b^{(m-1)})^{\alpha(s+1)}\right) \\
    & \nonumber \qquad \qquad \times \max\left(1,\log \log (b^{(m-1)}+1)\right)^{\alpha(1+\epsilon)}\\
     \nonumber
    &  \qquad + \cdots \\
    & \qquad  \left . +   b^{p(1-\alpha)}\max\left(1,(\log b^p)^{\alpha(s+1)}\right) \max\left(1,\log \log (b^p+1)\right)^{\alpha(1+\epsilon)}  \right\}.
\end{align}
Next, we factor out the largest logarithmic term and compute the sum of the geometric series: 
\begin{align*}
     P_\alpha(\vh,\vgamma,\lambda b^p,s)
     & \le \frac { (b-1) C_{P}(b,\alpha, \vgamma, \epsilon,s)}{\lambda b^p} \left\{ b^{p(1-\alpha)} \cdot \frac{(1 - b^{(m-p+1)(1- \alpha )})}{1 - b^{1- \alpha}}\right\} \\
    & \qquad  \times \max(1, (\log b^{m})^{\alpha(s+1)})) [\max(1,\log \log (
    b^m+1))]^{\alpha(1+\epsilon)} \\
    &\leq  \frac {(b-1)C_{P}(b,\alpha, \vgamma, \epsilon,s)}{\lambda b^{p\alpha}}  \frac{1 - b^{(m-p+1)(1- \alpha )}}{1 - b^{1- \alpha}} \\
    &\qquad \times \max(1, (\log b^{m})^{\alpha(s+1)})) [\max(1,\log \log (
    b^m+1))]^{\alpha(1+\epsilon)} \\
    &\le \frac{\lambda^{\alpha -1}{C}_{P}(b,\alpha, \vgamma,\epsilon,s)}{n^{\alpha}}  \frac {(b-1)} {1 - b^{1- \alpha}} \\
    &\qquad \times \max(1, (\log n )^{\alpha(s+1)})) [\max(1,\log \log (
    n+1))]^{\alpha(1+\epsilon)}   
\end{align*}

This completes the proof of \eqref{eq:mainresulttwo}. In this case of fixed $\lambda$, the upper bound decays nearly like $\Order(n^{- \alpha})$ as $n \rightarrow \infty$, i.e, $p \to \infty$,  
which is an improvement on \eqref{eq:mainresultone}. However, for $\lambda/n$ bounded away from zero as $n \to \infty$ the decay of $P_\alpha$ is the same as that in \eqref{eq:mainresultone}.

\subsection{The figure of merit $P_{\alpha}(\vh,\vgamma,\lowercase{n},s)$ when the coordinate weights satisfy summability condition }

\Cref{thm:one} makes no assumptions on the coordinate weights. In this section, we demonstrate how to eliminate the logarithmic terms in our upper bounds in \Cref{thm:one}---which grow exponentially as $s \to \infty$---by assuming a summability condition of the coordinate weights. In fact, we eliminate the $s$-dependence completely.

\begin{theorem} \label{thm:two}
Suppose $\alpha > 1$.
\begin{enumerate}\renewcommand{\labelenumi}{\alph{enumi}.}
    \item If $\sum_{j = 1}^{\infty} \gamma_j ^{\alpha/(1+2\delta)} < \infty$ for some $\delta$ satisfying $0 < \delta < (\alpha -1)/2$, then
\begin{equation} \label{eq:thmtwoa}
     P_\alpha(\vh,\vgamma,n, s)
        \leq \frac{\widetilde{C}_P(b,\alpha, \alpha/(1 + 2\delta), \vgamma, \delta)}{n} \frac{b - 1 }{1 - b^{-\delta}} \qquad
        \forall n, s \in \naturals.
\end{equation}

\item If $\sum_{j = 1}^{\infty} \gamma_j < \infty$ and $n = \lambda b^p$ for some positive integer $p$ and some positive integer $\lambda$ that is relatively prime with respect to $b$, then for any fixed $\delta $ such that $0 < \delta < \alpha - 1$:
\begin{equation*}
     P_\alpha(\vh,\vgamma,n,s)
        \leq \frac{\lambda^{\alpha - 1 - \delta }\widetilde{C}_P(b,\alpha, \vgamma, \delta)}{n^{\alpha - \delta}}\frac{(b-1)}{1 - b^{1 - \alpha + \delta}} \qquad 
        \forall  n, \lambda, s \in \naturals.
\end{equation*}
\end{enumerate}

\end{theorem}

\begin{proof}
By \cite[Theorem 2, case (iii)]{HicNie03a} we know that:
\begin{equation} \label{eq:nosbd}
      P_\alpha(\vh,\vgamma,b^m,s)
        \leq \widetilde{C}_P(\alpha, a, \vgamma, \delta) (b^m)^{-\alpha/a + \delta},
\end{equation}
for $a \in [1, \alpha]$, provided that  $\sum_{j = 1}^{\infty} \gamma_j^a < \infty$. There are two cases.

The first case is to recover the $\Order(n^{-1})$ decay as in \eqref{eq:mainresultoneb} but without the dimension dependence.  Choose a small positive $\delta$ that is less than $(\alpha -1)/2$, and let $a = \alpha/(1 + 2\delta)$.   
Substituting \eqref{eq:nosbd} into \eqref{eq:Palphan} we get: %with $a = \alpha$ we get: %and following the same argument as in \eqref{eq:mainresultone}, we get the following upper bound on $P_{\alpha}$: 

\begin{align*}
    \MoveEqLeft{P_\alpha(\vh, \vgamma, n, s)} \\
    &\leq \frac{(b - 1)}{n} \widetilde{C}_P(b,\alpha, \alpha/(1 + 2\delta), \vgamma, \delta)  \Bigl[
        b^{m(1-\alpha/a + \delta)} + b^{(m-1)(1-\alpha/a + \delta)} + \cdots \\
        & \qquad \qquad 
     + b^{0(1-\alpha/a+ \delta)} \Bigr] \\
    &\leq \frac{(b - 1)}{n} \widetilde{C}_P(b,\alpha, \alpha/(1 + 2\delta), \vgamma, \delta)  \Bigg[ 
        b^{-m \delta} + b^{-(m-1)\delta} + \cdots + 1 \Big] \\
    &\leq \frac{\widetilde{C}_P(b,\alpha, \alpha/(1 + 2\delta), \vgamma, \delta)}{n} \frac{b - 1 }{1 - b^{-\delta}}.
\end{align*}

For the case $n = \lambda b^p$, in \Cref{thm:one} we proved that the fastest decay we can achieve is nearly $\Order(n^{-\alpha})$. The weakest assumption we can make on the coordinate weights in this case to achieve this decay is  $a = 1$. Then, from \eqref{eq:nosbd}, the upper bound on $P_\alpha$ for $n = \lambda b^p$ becomes: 

\begin{align*}
        P_\alpha(\vh, \vgamma, \lambda b^p, s) 
    &\leq \frac{(b - 1)}{\lambda b^p} \widetilde{C}_P(b,\alpha, 1, \vgamma, \delta)  \Bigg\{ 
        b^{m(1-\alpha + \delta)} + b^{(m-1)(1-\alpha + \delta)} + \cdots \\
        & \qquad \qquad 
     + b^{p(1-\alpha + \delta)} \Bigg\} \\
    &= \frac{(b - 1)}{\lambda b^p} \widetilde{C}_P(b,\alpha, 1, 1, \vgamma, \delta) \Bigg\{ \frac{b^{p( 1 - \alpha + \delta)}(1 - b^{(m-p+1)(1-\alpha +\delta)})}{1 - b^{1 - \alpha + \delta}} \Bigg\} \\
    &= \frac{(b - 1)}{\lambda b^{p\alpha - \delta}} \widetilde{C}_P(b,\alpha, 1, \vgamma, \delta) \Bigg\{ \frac{1 - b^{(m-p+1)(1-\alpha+\delta)}}{1 - b^{1 - \alpha + \delta}} \Bigg\} \\
    &\leq \frac{\lambda^{\alpha - 1 - \delta }\widetilde{C}_P(b,\alpha, 1, \vgamma, \delta)}{n^{\alpha - \delta}} \frac{(b-1)}{1 - b^{1 - \alpha + \delta}}. 
\end{align*}
We emphasize that these upper bounds on $P_\alpha$ are independent of $s$.
\end{proof}

The theorems above establish a theoretical upper bound on the figure of merit $P_\alpha$ for $n$ other than $b^m$, showing a slower decay of $\Order(n^{-1})$ for arbitrary $n$. However, when $n = \lambda b^m$ for a fixed $\lambda$, the decay is similar to that for the preferred values $n = b^m$. In the next subsection, we study this decay numerically for specific parameter choices to evaluate how well the theoretical bounds hold in practice.

\subsection{Numerical Experiments}

In this subsection, we present numerical experiments illustrating the decay of the upper bound on $P_\alpha$ derived in previous subsections, with $\alpha = 2$. All lattice samples 
were generated using QMCPy \citep{QMCPy2026a} with the default generating vector $\vh$.

In \Cref{fig:p_alpha_O} we plot the upper bound on $P_2(\vh,[0.1,0.05],n,1\!\!:\!\!2)$ for all positive integers $n \leq 2^{10}$. The decay is nearly $\Order(n^{-2})$ when $n$ is a power of $2$, and closer to $\Order(n^{-1})$ for arbitrary $n$, as expected. For $n =\lambda 2^p$ with $\lambda = 1, 3, 5, 7$, \Cref{fig:/p_alpha_lambda_O} shows that the upper bound %on $P_2(\vh,[0.1, 0.05],\lambda 2^p,1:2)$  
decays nearly like $\Order(n^{-2})$ in all cases, as proven in \Cref{thm:one}, with generally smaller values for smaller $\lambda$.

\begin{figure}[H]
    \centering
    \includegraphics[width=0.6\linewidth]{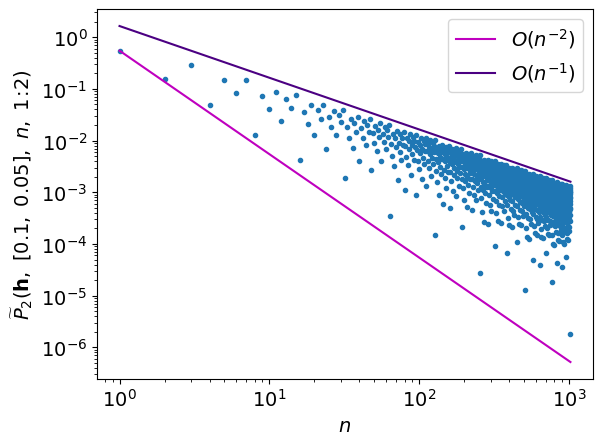}
    \caption{The decay of the upper bound on $P_2(\vh,[0.1, 0.05],n,1:2)$ for $ 1 \leq n \leq 2^{10}$.  %for the default $\vh$ in QMCPy. For values of $n$ that are powers of $2$, the decay is nearly $\Order(n^{-2})$, while for arbitrary $n$ it is closer to  $\Order(n^{-1})$. 
    }
    \label{fig:p_alpha_O}
\end{figure}

\begin{figure}[H]
    \centering
    \includegraphics[width=0.6\linewidth]{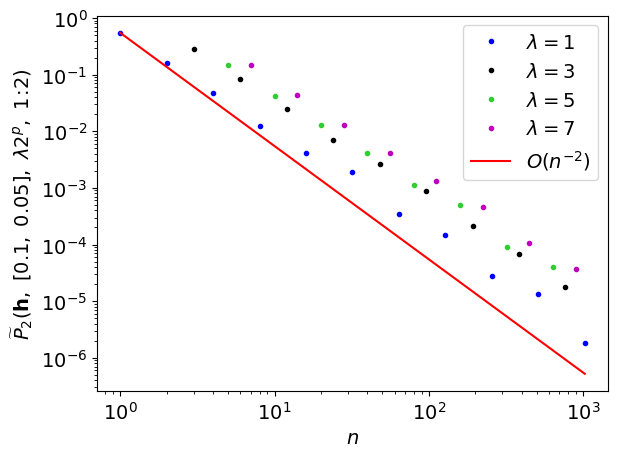}
    \caption{The decay of the upper bound on $P_2(\vh,[0.1, 0.05],\lambda 2^p,1:2)$ for $ 1 \leq \lambda 2^p \leq 2^{10}$. 
    %for the default $\vh$ in QMCPy. Note the $\Order(n^{-2})$ decay in all cases.  The values are generally smaller for smaller $\lambda$.
    }
    \label{fig:/p_alpha_lambda_O}
\end{figure}

% We now explore how lattices with arbitrary $n$ perform in a numerical integration example. 
% to complement 
To further explore the performance of lattices for  arbitrary $n$, we consider a numerical integration example using the $s$-dimensional Keister integral \citep{Kei96},
\begin{equation*}
    \mu = \int_{\reals^{s}} \cos(\|\vx\|_2) \exp{(-\|\vx\|_2^2}) \, \dif \vx = \mathbb{E}[f(\vX)]
\end{equation*}
where $f(\vx) = \pi^{s/2}\cos(\|\vx\|_2)$ for all $\vx \in \reals^s$ and $\vX \sim \mathcal{N}(\vzero_s, \mathbf{I}_{s\times s}/2)$. Then
\[
\mu = \int_{[0,1]^{s}} \cos\left(\sqrt{\frac{1}{2}\sum_{j=1}^{s}\Phi^{-1}(z_j)^2} \ \right) \, \dif \vz,
\]
where $\Phi$ is the standard normal cumulative distribution function, and the last equality follows substituting $x_j = \Phi^{-1}(z_j)/\sqrt{2}$ for each $j = 1, \ldots, s$. %inverve transform sampling

  We consider $s = 6$ and approximate $\mu$ using the sample mean over the first  %$n = 1, 2, \ldots, 2^{15}$ 
 $n$ points of a shifted lattice. In \Cref{fig:keister_n} we plot absolute errors for all positive integers $n \leq 2^{15}$, while \Cref{fig:keister_lambda} focuses  on $n = \lambda 2^p$ with the same values of $\lambda$ as in \Cref{fig:/p_alpha_lambda_O}.

\begin{figure}[H]
    \centering
    \includegraphics[width=0.6\linewidth]{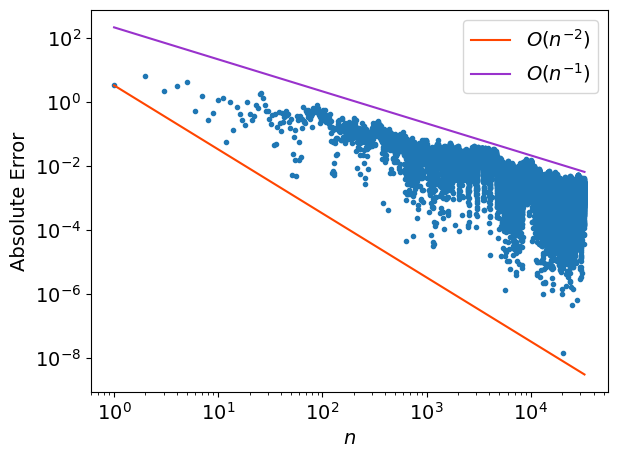}
    \caption{The decay of the absolute error in approximating six-dimensional Keister integrand by a sample mean over the first $1 \le n \le 2^{15}$ points. }
    \label{fig:keister_n}
\end{figure}

\begin{figure}[H]
    \centering
    \includegraphics[width=0.6\linewidth]{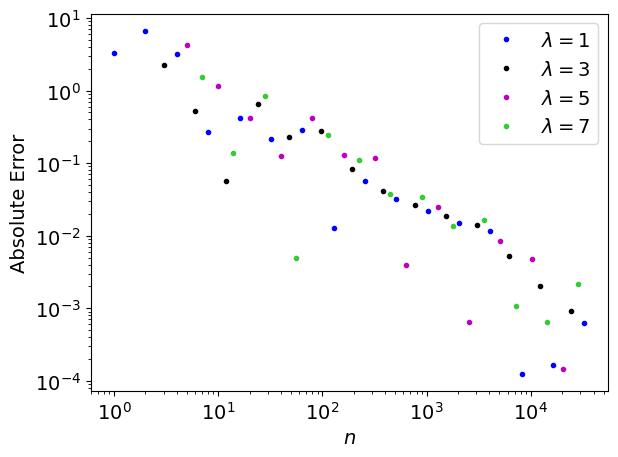}
    \caption{The decay of the absolute error in approximating 6-dimensional Keister integrand by a sample mean over the first $1 \le \lambda 2^p \le 2^{15}$ points. %The errors are generally smaller for smaller $\lambda$.
    }
    \label{fig:keister_lambda}
\end{figure}

So far, we have analyzed the figure of merit $P_\alpha$ for $n$ other than $b^m$ in this Banach space setting. In the next section, we consider a related Hilbert space and introduce unequal sample weights in our cubature formula. We also investigate a quantity related to $P_\alpha$, which allows these unequal weights. 

\section{The Hilbert space setting with $P_{\alpha,2}$ and optimal sample weights}
In \Cref{sec:Palphabd} we only derived \emph{upper bounds} on the values of $P_{\alpha}$ for arbitrary $n$.  The exact formula for $P_{\alpha}$ for arbitrary $n$ involves an infinite sum over the wavenumbers that cannot be reduced to a finite sum.  However, if we switch from the Banach space, $\mathcal{F}_\alpha$, defined in \eqref{eq:Banachspace}, to a Hilbert space, $\mathcal{F}_{\alpha,2}$, then we can evaluate the corresponding figure of merit, $P_{\alpha,2}$ for arbitrary values of $n$ in terms of a finite sum.  
%add something before 

We define a Hilbert space of functions, related to the Banach space defined in \eqref{eq:Banachspace} by taking the weighted $\ell_2$ norm of the Fourier coefficients rather than the $\ell_\infty$ norm: 
\begin{gather*}
    \cf_{\alpha,2} = \{ f \in \cl_2[0,1)^{s} :
\norm[\cf_{\alpha,2}]{f} < \infty \}, \\
\norm[\cf_{\alpha,2}]{f} := \left[\sum_{\vk \in \integers^s} \left(\tr(\vk,\vgamma)^{\alpha}\abs{\tilde{f}(\vk)}\right)^2 \right]^{1/2} .
\end{gather*}

Note that this is a stronger norm, i.e., $\norm[\cf_{\alpha}]{f} \le \norm[\cf_{\alpha,2}]{f}$ for all $f \in \cf_{\alpha,2}$.  As a consequence
\begin{multline} \label{Pa2Parelation}
P_{\alpha,2}(\vh,\vgamma,n,s) : = \sup_{\norm[F_{\alpha,2}]{f} \le 1 } \abs{\int_{[0,1)^{s}} f(\vx) \, \dif \vx - \appxint_n(f)} \\
\le \sup_{\norm[F_{\alpha}]{f} \le 1 } \abs{\int_{[0,1)^{s}} f(\vx) \, \dif \vx - \appxint_n(f)} = P_{\alpha}(\vh,\vgamma,n,s).
\end{multline}

\subsection{Derivation of exact formulae involving arbitrary sample weights}
Before proceeding with the derivation of $P_{\alpha,2}$, we generalize to arbitrary sample weights for cubature, i.e., 
\[
\appxintn(f) =  \sum_{i=0}^{n-1} w_if(\vz_i).
\]
Then the tight bound on the error analogous to \eqref{eq:wcerrPalpha} becomes: 
\begin{align*}
\nonumber
\MoveEqLeft{\sup_{\norm[\cf_{\alpha,2}]{f} \le 1} \abs{\int_{[0,1)^{s}} f(\vx) \, \dif \vx - \appxint_n(f)}^2} \\
& = \sup_{\norm[\cf_{\alpha,2}]{f} \le 1} \sum_{\vk \in \integers^s \setminus\{\vzero\}} \abs{\appxint_n(\me^{2 \pi \sqrt{-1} \vk^T \cdot})}^2\tr(\vk,\vgamma)^{-2\alpha} \\
&=\sum_{\vk \in \integers^s \setminus\{\vzero\}} \sum_{i,j = 0}^{n-1} \me^{2 \pi \sqrt{-1} (\vx_i - \vx_j)} w_i w_j \tr(\vk,\vgamma)^{-2\alpha} \\
&= \underbrace{ -1 + \sum_{i,j = 0}^{n-1} w_i w_j \sum_{\vk \in \integers^s} \me^{2 \pi \sqrt{-1} (\vx_i - \vx_j)} \tr(\vk,\vgamma)^{-2\alpha}}_{P_{\alpha,2}^2(\vh,\vgamma,n,1:s, \vw)},
\end{align*}
where $\vw = \bigl( w_i \bigr)_{i=0}^{n-1}$.

Note that for lattices, i.e., $n=b^m$, and for equal weights $\vw = b^{-m} \vone $, it follows that $P_{\alpha,2}^2(\vh,\vgamma,b^m,s, b^{-m} \vone) = P_{2\alpha}(\vh,\vgamma,b^m,s)$ by definition.

For positive integer $\alpha$, we may again appeal to Bernoulli polynomials as we did in \eqref{eq:P2alpha}:
\begin{align*}
\nonumber
P_{\alpha,2}^2(\vh,\vgamma,n,s, \vw)
& = -1 + \sum_{i,j = 0}^{n-1} w_i w_j  
\prod_{\ell = 1}^s [1 + \tilde{\gamma}_{\ell}^{2\alpha}B_{2\alpha}(x_{i,\ell} - x_{j, \ell}\bmod 1)] \\
& = - 1 + \vw^T \mK \vw, \\
\text{where } \mK  & = \biggl( \prod_{\ell = 1}^s [1 + \tilde{\gamma}_{\ell}^{2 \alpha}B_{2\alpha}(x_{i,\ell} - x_{j, \ell} \bmod 1 )]\biggr)_{i,j=0}^{n-1} ,
\end{align*}
and $\tilde{\gamma}_{\ell}$ was defined in \eqref{eq:gammatilde}.

To find the optimal sample weights $\vw$, we solve the following minimization problem:
\[ \min_{\vw} \vw^T \mK \vw \quad \text{s.t.} \quad \vw^T \vone = 1, \]
where the constraint ensures that constant functions are integrated exactly. This problem has a known solution but we derive it for the sake of completeness. We solve it using the method of Lagrange multipliers. 
First, we define the Lagrangian:  
\[
\mathcal{L}(\vw, \lambda) = \vw^T \mK \vw - \lambda (\vw^T \vone - 1).
\]
We compute the gradient of \(\mathcal{L}\) and set it equal to zero:
\[
\nabla_{\vw} \mathcal{L} = 2 \mK \vw - \lambda \vone = 0 \implies 2\mK \vw = \lambda \vone.
\]
\[
\nabla_{\lambda} \mathcal{L} = -(\vw^T \vone - 1) = 0 \implies \vw^T \vone = 1.
\]
Since $\mK$ is positive-definite, we know it is a minimum. 
From the first equation above we get:
\[
\vw = \frac{\lambda}{2} \mK^{-1} \vone.
\]
Substituting this expression for $\vw$ into the constraint to find the value of $\lambda$:
\[
\left(\frac{\lambda}{2} \mK^{-1} \vone \right)^T \vone = 1
\implies \frac{\lambda}{2} \vone^T \mK^{-1} \textbf{1} = 1
 \implies \lambda = \frac{2}{\vone^T \mK^{-1} \vone}.
 \]
Thus, the optimal sample weights are: 
\[
\vw_{\opt} := \frac{\mK^{-1}\vone}{\vone^T \mK^{-1}\vone }.
\]

The cost of solving for the weights $\vw$ is $\Order(sn^2)$ to evaluate the kernel and $\Order(n^3)$ to invert the Gram matrix, $\mK$.  (We assume that evaluating one element of the Gram matrix is proportional to the dimension.)  However, since $\mK$ is Toeplitz, we can reduce the cost of solving for optimal sample weights to $\Order(sn+n^2)$.   %$\Order(sn) + \Order(n^2)$.

\subsection{Error bounds in the Hilbert space setting}

\begin{theorem} \label{thm:three}
The upper bounds derived in \Cref{thm:one} and \Cref{thm:two} also hold for $P_{\alpha,2}$.
\end{theorem}

\begin{proof}
    This follows immediately by the inequality in \eqref{Pa2Parelation}.
\end{proof}

Corollary 1 in \cite{HicEtal10a} states that if the worst case error bound for a Hilbert space like those we consider here is $Cb^{-m \alpha})$ for $\alpha \ge 1$ for the preferred sample size, $n = 1, b, b^2, \ldots$, then for arbitrary sample size it is $C n^{-1} \log(n)$. In  \Cref{thm:three} our hypothesis is not on the error bound but on the parameter $\alpha$ defining the Hilbert space.  The corresponding result in \eqref{eq:thmtwoa} avoids the extra $\log(n)$ factor for $\alpha > 1$, granted by inflating the leading constant.

\subsection{Numerical results}

To illustrate how optimal weights affect the behavior of $P_{\alpha,2}$ we present several numerical results. Figure \ref{fig:ssdisc-vs-ssdiscopt} shows the decay of $P_{1,2}$ under optimal sample weights and under equal weights for $n \leq 2^{12}$. In both cases, the decay is nearly $\Order(n^{-1})$. However, the values of $P_{1,2}$ are consistently smaller and non-increasing for optimal weights. 
% To illustrate the advantage of using optimally chosen sample weights 
Figures \Cref{fig:ratio-floor} and \ref{fig:ratio-ceil} compare $P_{1,2}$ at each positive integer $n \leq 2^7$ to its value at the nearest power-of-two sample sizes. %that are powers of 2. 
Figure \ref{fig:ratio-floor} uses the largest such sample size no greater than $n$ % that is a power of 2 and no greater than $n$, 
and shows that, with optimal sample weights the ratios remain at or below one as $n$ increases, while the ratios with equal weights fluctuate more and frequently exceed one. Figure \ref{fig:ratio-ceil} uses the smallest power-of-two sample size greater than or equal to $n$. %that is a power of $2$ and no less than $n$, 
 Although in both cases the ratios are greater than one, optimal weighs produce lower ratios than the equal weights.

% Figures \ref{fig:ratio-floor} and \ref{fig:ratio-ceil} compare $P_{1,2}$ at each $n$ to its value at the largest sample size that is a power of 2 and no greater than $n$ and the smallest sample size that is a power of $2$ and no less than $n$, respectively. 

% Next we want to numerically study the quality of $P_{1,2}$ under optimally chosen sample weights and compare it to the the quality of $P_{1,2}$ under equal  sample weights.

\begin{figure}[H]
    \centering
    \includegraphics[width=0.6\linewidth]{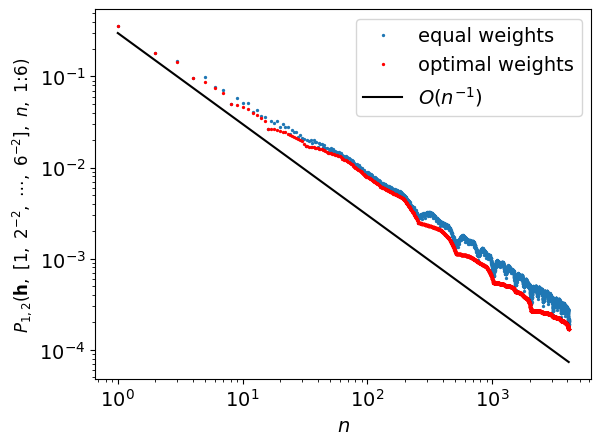}
    \caption{The decay of $P_{1,2}(\vh, [1, 2^{-2}, \cdots, 6^{-2}], n, 1:6)$ with equal sample weights vs unequal sample weights for $ 1 \leq n \leq 2^{12}$ for the default $\vh$ in QMCPy. The quantity $P_{1,2}$ decays nearly like $\Order(n^{-1})$, and for optimal weights it is non-increasing.}
    \label{fig:ssdisc-vs-ssdiscopt}
\end{figure}

\begin{figure}[H]
    \centering
    \includegraphics[width=0.6\linewidth]{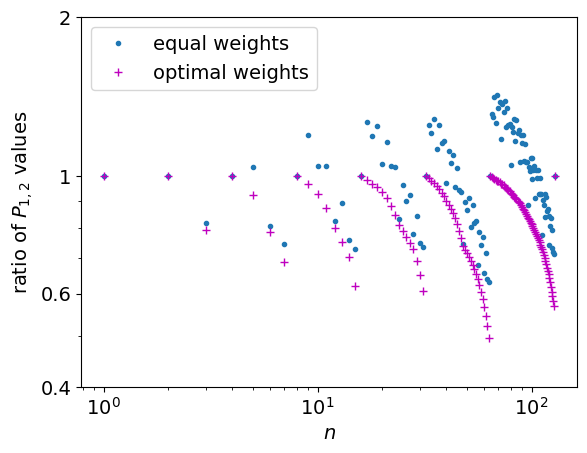}
    \caption{The ratio $P_{1,2}(\vh,[1, 1/4],n,1:2)/P_{1,2} (\vh,[1, 1/4],2^{\lfloor \log_2(n) \rfloor},1:2)$ for $ 1 \leq n \leq 2^7$. This shows how much better (smaller than one) or worse (greater than one) the figure of merit is compared to the largest sample size that is a power of $2$ and no greater than $n$. }
    \label{fig:ratio-floor}
\end{figure}

\begin{figure}[H]
    \centering
    \includegraphics[width=0.6\linewidth]{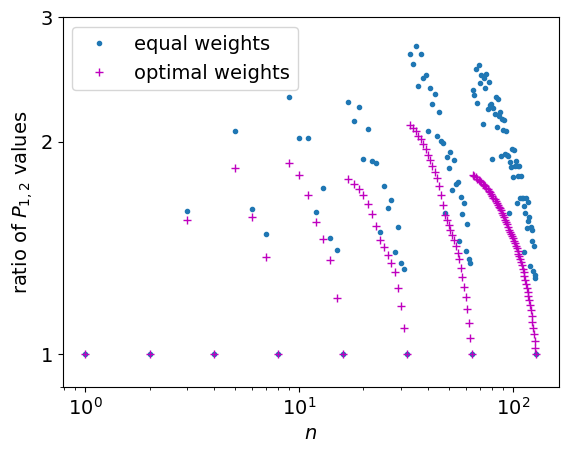}
    \caption{The ratio $P_{1,2}(\vh,[1, 1/4],n,1:2)/P_{1,2} (\vh,[1, 1/4],2^{\lceil \log_2(n) \rceil},1:2)$ for $ 1 \leq n \leq 2^7$. This shows how much worse (greater than one) the figure of merit is compared to the smallest sample size that is a power of $2$ and no less than $n$.}
    \label{fig:ratio-ceil}
\end{figure}
% These numerical results illustrate the advantage of using optimally chosen sample weights as opposed to equal weights. In Figure \ref{fig:ssdisc-vs-ssdiscopt} we note that $P_{1,2}$ decays nearly like $\Order(n^{-1})$ in both cases, but the values are consistently smaller and non-increasing for optimal weights. Figure \ref{fig:ratio-floor} compares 
% $P_{1,2}$ at each $n$ to its value at the largest power of 2 and no greater than $n$. For optimal sample weights, the ratios do not exceed one as $n$ increases, while the ratios for equal weights demonstrate more fluctuations and are often greater than one. Figure \ref{fig:ratio-ceil} compares $P_{1,2}$ at each $n$ to its value at the smallest power of 2 and no less  than $n$. Although in both cases the ratios are generally above one, optimal weighs produce lower values than the equal weighs. 
% 

\section{Discussion and Future Work} %explain more what it all means 

The results in this paper illustrate that the figure of merit $P_\alpha$ does have a worse upper bound for general $n$ compared to the preferred $n = b^m$ for a fixed prime $b$, as expected. However, this is not a handicap if the smoothness of the integrand is limited.  Moreover, we can still attain the optimal decay of $\Order(n^{-\alpha})$ for $n =\lambda b^p$ for a fixed integer $\lambda$ and $p \to \infty$. This suggests that even if we cannot obtain a full lattice due to early termination or computational time budget, sample sizes of the form $n = \lambda b^p$ could still be useful in practice.

To investigate whether unequal sample weights can improve the cubature error for arbitrary $n$, we analyzed a related quantity $P_{\alpha,2}$. Numerical experiments with $\alpha=1$ 
demonstrate that by using optimal sample weights, even when the next full lattice is not attainable, we obtain a smaller cubature error for arbitrary $n$ compared to the full lattice with sample size $b^m < n$. 

  %Comparing  $P_{1,2}$ at each $n$ to $P_{1,2}$ for the previous full lattice 
 %This suggests that we can obtain (pretty small) cubature error if we cannot have the next full lattice, provided that we use optimal weights. 

Future work could explore the performance of lattices with arbitrary $n$ for concrete numerical test cases, such as option pricing or integration problems described in \cite[Integration]{simulationlib}. Additionally, further numerical experiments could investigate the decay of the upper bound on $P_\alpha$ and the decay of $P_{\alpha,2}$ for a wider range of $\alpha$ values.

% \left\{\right\}
\bibliographystyle{elsarticle-harv.bst}
\bibliography{FJH26,FJHown26,LarysaReferences}
\end{document}